\documentclass[runningheads]{llncs}
\usepackage{graphicx}
\usepackage[vlined,linesnumbered,ruled]{algorithm2e}
\usepackage{amsmath}
\usepackage{hyperref}
\hypersetup{
    colorlinks=true,
    linkcolor=blue,
    filecolor=magenta,      
    urlcolor=cyan,
}
\begin{document}
\title{A Makespan and Energy-Aware Scheduling Algorithm for
Workflows under Reliability Constraint on a Multiprocessor
Platform}
\titlerunning{Makespan, Energy Effective Reliability Constrained Scheduler}
%
\author{Atharva Tekawade  \and
 Suman Banerjee }
\authorrunning{Tekawade and Banerjee}
%
\institute{Department of Computer Science and Engineering, \\ Indian Institute of Technology Jammu, Jammu 181221, India.\\
\email{\{2018uee0137,suman.banerjee\}@iitjammu.ac.in}}
\maketitle              
\begin{abstract}
Many scientific workflows can be modeled as a Directed Acyclic Graph (henceforth mentioned as DAG) where the nodes represent individual tasks, and the directed edges represent data and control flow dependency between two tasks. Due to the large volume of data, multi-processor systems are often used to execute these workflows. Hence, scheduling the tasks of a workflow to achieve certain goals (such as minimizing the makespan, energy, or maximizing reliability, processor utilization, etc.) remains an active area of research in embedded systems. In this paper, we propose a workflow scheduling algorithm to minimize the makespan and energy for a given reliability constraint. If the reliability constraint is higher, we further propose \textbf{E}nergy \textbf{A}ware \textbf{F}ault \textbf{T}olerant \textbf{S}cheduling (henceforth mentioned as EAFTS) based on active replication. Additionally, given that the allocation of task nodes to processors is known, we develop a frequency allocation algorithm that assigns frequencies to the processors. Mathematically we show that our algorithms can work for any satisfiable reliability constraint. We analyze the proposed solution approaches to understand their time requirements. Experiments with real-world Workflows show that our algorithms, MERT and EAFTS, outperform the state-of-art approaches. In particular, we observe that MERT gives 3.12\% lesser energy consumption and 14.14\% lesser makespan on average. In the fault-tolerant setting, our method EAFTS gives 11.11\% lesser energy consumption on average when compared with the state-of-art approaches. 
\end{abstract}
%
%
%
\section{Introduction}

Many real-time embedded applications, in the domains of automotive \cite{bolchini2012reliability}, avionics \cite{sahner1987performance}, industrial automation \cite{yang2020task}, and communication networks \cite{hassija2020dagiov} can be modeled as a directed acyclic graph where each node represents an individual task, and a directed edge represents a dependency relationship (either data or control or both) between two tasks. In the literature, they have also been referred to as the \emph{task graph} or \emph{precedence-constrained task graph}. A real-time system is characterized by its inherent ability to respond within a given stimulated time for any kind of external event. Traditionally, these workflows are executed using a multiprocessor system where the individual processing elements may be heterogeneous. In a heterogeneous system, the same task may exhibit very different characteristics depending upon which processing element (e.g. CPU, GPU, DSPU, etc.) is being used. Due to continuous demand for better performance and reliability subject to energy consumption constraints, there is a recent trend to use complex heterogeneous and distributed platforms for the execution of workflows. Now, due to the increasing level of heterogeneity, maintaining the reliability requirement becomes challenging. 
\par Scheduling of jobs in multiprocessor systems has got two main variations. The first one is offline scheduling (also known as static scheduling), where the details of the tasks (e.g. the processing time, nature of the job, etc.) is known as apriori. The other one is online scheduling, where the details of the tasks come up dynamically. However, in the case of real-world safety-critical systems, it is natural that all the timings are known in advance so that the safety critical criteria can be maintained. Hence, in this study, also we consider static scheduling. In general, the scheduling of the tasks modeled as a DAG is NP-Complete \cite{bruno1976computer}. So the key research focus remains to develop efficient heuristic solution approaches.
\par Due to the inherent level of heterogeneity (both at the hardware and network level) in a multiprocessor system, maintaining reliability for the execution of the task graph is one of the main criteria. Also, as mentioned previously, there is an inherent dependency between two tasks present in the task graph. Hence the failure of one task may even lead to the failure of the entire workflow. So, to maintain highly reliable execution, it is important to maintain the reliability of all the individual tasks of the workflow. Another important criterion of a scheduling algorithm is energy consumption. In particular, this becomes an important criterion when the embedded system becomes a wireless device such as a robot or a drone. In the literature, there exists a significant amount of study on scheduling algorithms for workflows on heterogeneous multiprocessor and cloud systems \cite{zhao2010maximizing,tang2012hierarchical,chang2014battery,ma2018online,dogan2002matching}. Among the existing solution approaches HEFT Algorithm is a popular one and leads to minimum makespan value in many cases \cite{topcuoglu2002performance}. Also, several other scheduling algorithms consider other parameters such as energy consumption, such as Least Energy Cost (LEC) \cite{xie2017energy} and reliability, such as Maximum Reliability (MR) \cite{xie2017minimizing}. Some studies consider all three important parameters \emph{i.e.} makespan, energy, and reliability. \cite{zhao2010maximizing,tang2012hierarchical}.

\par In this paper, we make the following contributions:
\begin{itemize}
    \item We propose MERT, a non-fault tolerant scheduling algorithm that minimizes makespan and energy under a given reliability constraint by allocating task nodes to processors based on finish time, execution time, and energy consumption depending on the wait time of a task.
    \item For higher reliability constraints, we propose EAFTS a fault-tolerant scheduling algorithm that minimizes energy consumption under a given reliability constraint by allocating each task node to processors with the least energy consumption. 
    \item Mathematically, we show that both MERT and EAFTS can achieve any reliability constraint in non-fault tolerant and fault-tolerant settings, respectively.
    \item Given an allocation of task nodes to processors, we come up with a frequency allocation algorithm that minimizes energy consumption under a reliability constraint in the fault-tolerant setting.
    \item We perform extensive simulation experiments on real-world task graphs and compare our methods with the state-of-art approaches.
\end{itemize}

The rest of the paper is organized as follows. Section \ref{Sec:SM} describes the system's model. Section \ref{Sec:PD} describes the problems we try to solve. The proposed solutions have been described in Section \ref{Sec:PA}. Section \ref{Sec:EXP} contains the experimental evaluation of the proposed solution methodology. Finally, Section \ref{Sec:CFA} concludes our study and gives future research directions.

\section{System's Model} \label{Sec:SM}
In this section, we describe the system's model and describe our problem formally. For any positive integer $n$, $[n]$ denotes the set $\{1, 2, \ldots, n\}$. Initially, we start by describing a multiprocessor system.

\subsection{Multiprocessor Platform}
Our model comprises m different heterogeneous processors $\mathcal{P} = \{u_k: k \in [m]\}$, where each pair can communicate data with each other. Each processor runs in a range of frequencies. Let $f_{u_k, \min}$ and $f_{u_k, \max}$ denote the lower and upper bounds of frequency respectively for $u_k$. For simplicity, we normalize the maximum frequency of each processor to one  \emph{i.e.} $f_{u_k, \max} = 1.$ Rest of the processor-specific parameters will be discussed in the subsequent sub-sections.

\subsection{Workflow}
A scientific workflow can be modeled as a DAG denoted by $G(V, E)$, where $V$ denotes the set of vertices: $\{v_i : i \in [n]\}$ and each task node $v_i$, represents a task of the application. $E = \{(v_i, v_j) : v_i, v_j \in V\}$ denotes the set of edges of our task graph. A directed edge between task node $v_i$ to $v_j$ indicates a precedence relationship between $v_i$ and $v_j$ \emph{i.e.} $v_j$ cannot start unless it has received the necessary output data from $v_i$.
The weight of the edge $(v_i, v_j)$ denoted by $w(v_i, v_j)$ gives an idea of the communication time between the two tasks. Below, we present some terminology that will be used subsequently.
\begin{definition} \label{def1}
$pred(v_i)$ denotes the set of immediate predecessor nodes of task node $v_i$. Mathematically, $pred(v_i) = \{v_j : (v_j, v_i) \in E \} $
\end{definition}  
\begin{definition} \label{def2}
$succ(v_i)$ denotes the set of immediate successor nodes of task node $v_i$. Mathematically, $succ(v_i) = \{v_j : (v_i, v_j) \in E \} $
\end{definition}
\begin{definition} \label{def3}
$v_{entry}$ denotes the entry task. It is a redundant node added to have a proper notion of the first task. A directed edge with zero weight is added from $v_{entry}$ to every $v$ \emph{s.t.} $pred(v) = \emptyset$.
\end{definition}
\begin{definition} \label{def4}
$v_{exit}$ denotes the exit task. It is a redundant node added to have a proper notion of the last task. A directed edge with zero weight is added from every $v$ \emph{s.t.} $succ(v) = \emptyset$ to $v_{exit}$.
\end{definition}
For simplicity assume $v_{entry} = v_1$ and $v_{exit} = v_n$.

\subsection{Timing Metrics}
This sub-section discusses the start, execution, and finish times associated with executing a task on a processor. Let $T_s[v_i, u_k]$ and $T_f[v_i, u_k]$ denote the start and finish times of executing $v_i$ on $u_k$. Let $T_{exec}[v_i, u_k]$ denote the execution time of executing $v_i$ on $u_k$ at the maximum frequency, which can be determined through WCET analysis method (Worst Case Execution Time) \cite{xie2016resource}. Running the processor at a lower frequency $f$ increases the time proportionately. Once a task starts on a processor, it will run to completion \emph{i.e.} we do not assume pre-emption. This leads us to Equation No. \ref{eqn1}.

\begin{equation} \label{eqn1}
    T_f[v_i, u_k] = T_s[v_i, u_k] + \frac{T_{exec}[v_i, u_k]}{f}
\end{equation}

Due to the task dependencies, a task must transfer data to its successor nodes. Let the time required to communicate the data between tasks $v_i$ to $v_j$  be denoted by $T_{comm}[v_i, v_j]$. As mentioned earlier, the communication time depends on the weight of the edge connecting $v_i$ to $v_j$. Additionally, the communication time is negligible if both tasks are allocated to the same processor. This leads us to Equation No. \ref{eqn2}.

\begin{equation} \label{eqn2}
\scriptsize
T_{comm}[v_i, v_j]= 
\begin{cases}
    w(v_i, v_j),& \text{if $v_i, v_j$ are scheduled on different processors}\\
    0,              & \text{otherwise}
\end{cases}
\end{equation}

A task can start on a processor provided both the below conditions hold:
\begin{itemize}
    \item It has received output from all of its predecessor task nodes.
    \item The processor on which it is scheduled is not executing another task at that time.
\end{itemize}

This leads us to the below equation:
\begin{equation} \label{eqn3}
\scriptsize
T_s[v_i, u_k] = 
\begin{cases}
    \max \{avail[u_k], \max_{v_j \in pred(v_i)} \{T_f[v_j, u_{k'}] + T_{comm}[v_j, v_i] \} \},& \text{if $v_i \neq v_1$}\\
    0,              & \text{otherwise}
\end{cases}
\end{equation}

where $avail[u_k]$ denotes the earliest time that the processor $u_k$ is free after executing its previous task, and $v_j$ is assumed to be scheduled on processor $u_{k'}$.  
The makespan or schedule length is defined as the time when the last task finishes executing.

\subsection{Energy}
The power consumption of a processor consists of frequency-dependent dynamic consumption, frequency-independent dynamic consumption, and static consumption components \cite{huang2020dynamic}. The frequency dependant dynamic component is the dominant one and can be written as:

\begin{equation} \label{eqn4}
P = \gamma \cdot c \cdot v^2 \cdot f
\end{equation}

where $\gamma$ is the activity factor, $c$ is the loading capacitance, $v$ is the supply voltage, and $f$ is the operating frequency. Since $f \propto v$, we see that $P \propto f^{\alpha}$. $P_{u_k}$ denotes the sum of the frequency-independent dynamic consumption and static consumption components.

Let $P_{u_k}^{f}$ denote the overall power consumed by $u_k$ when operated at frequency $f$, given by the Equation No. \ref{eqn5}.
\begin{equation} \label{eqn5}
P_{u_k}^f = P_{u_k} + c_{u_k} \cdot f^{\alpha_{u_k}}
\end{equation}

where $c_{u_k}, \alpha_{u_k}$ denote the processor constants for $u_k$. \\

Energy consumed will be obtained by taking the product of power and execution time, as shown below.

\begin{equation} \label{eqn6}
E_{v_i, u_k}^{f} = P_{u_k}^{f} \cdot \frac{T_{exec}[v_i, u_k]}{f}
\end{equation}

where $E_{v_i, u_k}^{f}$ denotes the energy consumption when the task $v_i$ is executed on $u_k$ with frequency $f$.

\subsection{Reliability}
As in many other works \cite{huang2020dynamic}, \cite{xie2017energy} we study dominant transient faults related to processor frequency and can be modeled by the below exponential distribution: 

\begin{equation} \label{eqn7}
\lambda_{u_k}(f) = \lambda_{u_k} \cdot 10^\frac{d_{u_k}(1-f)}{1-f_{u_k, \min}}
\end{equation}

where $\lambda_{u_k}$ denotes the average number of faults per second at the maximum frequency and $d_{u_k}$ is a processor constant. \\

The reliability is modeled using a Poisson distribution, with parameter $\lambda_{u_k}(f)$. $\mathcal{R}_{v_i, u_k}^{f}$ denotes the reliability when the task $v_i$ is executed on the processor $u_k$ with frequency $f$ and this is given in Equation No. \ref{eqn8}.

\begin{equation} \label{eqn8}
\mathcal{R}_{v_i, u_k}^{f} = e^{-\lambda_{u_k}(f) \cdot \frac{T_{exec}[v_i, u_k]}{f}}
\end{equation}

\section{Problem Definition} \label{Sec:PD}
A schedule $\mathcal{S} = \{$\textbf{k},\textbf{f}$\}$ is defined by a task to processor(s) mapping vector (\textbf{k}) and a processor to frequency allocation vector (\textbf{f}). Assume $v_i$ is scheduled on processors: $\{u_{ij} : j \in [i_k]\}$ with frequencies $\{f_{ij} : j \in [i_k]\}$. Consequently, the vectors \textbf{k}, \textbf{f} are defined so that the $i^{th}$ element for each denotes the set of allocated processors and frequencies respectively. The reliability of task $v_i$ is given by the below equation \cite{xie2017energy}:

\begin{equation} \label{eqn9}
    \mathcal{R}_{v_i, \textbf{k}[i]}^{\textbf{f}[i]} = 1 - \prod_{j=1}^{i_k} (1 - \mathcal{R}_{v_i, u_{ij}}^{f_{ij}})
\end{equation}

The task graph executes successfully when all the tasks execute successfully. Assuming that failures are independent, the reliability of the task graph is given by the product of the reliability of all the tasks as in Equation No. \ref{eqn10}.

\begin{equation} \label{eqn10}
\mathcal{R}(\textbf{k}, \textbf{f}) = \prod_{i=1}^{n} \mathcal{R}_{v_i, \textbf{k}[i]}^{\textbf{f}[i]}
\end{equation}

The energy consumption for task $v_i$ is given by Equation No. \ref{eqn11}.

\begin{equation} \label{eqn11}
E_{v_i, \textbf{k}[i]}^{\textbf{f}[i]} = \sum_{j=1}^{i_k} E_{v_i, u_{ij}}^{f_{ij}}
\end{equation}

The total energy consumption is given by the sum of individual energy consumption for each task, as shown in Equation No. \ref{eqn12}.
\begin{equation} \label{eqn12}
    E(\textbf{k}, \textbf{f}) = \sum_{i=1}^{n} E_{v_i, \textbf{k}[i]}^{\textbf{f}[i]}
\end{equation}

From Equation No. \ref{eqn7} and \ref{eqn8}, we can easily observe that both $\lambda_{u_k}(f)$ and $\mathcal{R}_{v_i, u_k}^{f}$ are increasing functions of $f$. Hence for maximum reliability, each processor must run at $f_{u_k, \max} = 1$. If we assume that each task has to be allocated to one processor (non-fault tolerant case), then for each task $v_i$ there is a processor with minimum value $-\lambda_{u_k} \cdot T_{exec}[v_i, u_k]$ that gives maximum reliability value $\mathcal{R}_{v_i, \max}^{\text{non-fault tolerant}}$. The corresponding maximum reliability for the workflow is denoted by $\mathcal{R}_{\max}^{\text{non-fault tolerant}}$.

The maximum reliability denoted by $\mathcal{R}_{\max}^{\text{fault-tolerant}}$ in the fault-tolerant setting is when each task is scheduled to run on all processors \emph{i.e.} \textbf{k}$[i] = \mathcal{P}$ and \textbf{f}$[i] = \{f_{u_k, \max} : k \in [m]\}$. Let the corresponding reliability value for task $v_i$ be $\mathcal{R}_{v_i, \max}^{\text{fault-tolerant}}$\\
Given a reliability constraint $\mathcal{R}_{req}$, the following scenarios can occur:

\renewcommand{\theenumi}{\roman{enumi}}%
\begin{enumerate}
\item $R_{req} \leq \mathcal{R}_{\max}^{\text{non-fault tolerant}}$: In this case, we focus on non-fault tolerant scheduling by assigning each task to a unique processor. The problem we study here is to minimize makespan and energy consumption. We name this as the non-fault tolerant setting.
\item $\mathcal{R}_{\max}^{\text{non-fault tolerant}} < \mathcal{R}_{req} \leq  \mathcal{R}_{\max}^{\text{fault-tolerant}}$: In this case, we focus on fault-tolerant scheduling by assigning each task to multiple processors. The problem we study here is that of minimizing energy consumption. We name this the fault-tolerant setting.
\item $\mathcal{R}_{\max}^{\text{fault-tolerant}} < \mathcal{R}_{req}$: No possible allocation can satisfy the given constraint in this case.
\end{enumerate}

\section{Proposed Solution Approach} \label{Sec:PA}
We follow a list-based scheduling strategy for both problem settings, which consists of two phases: Task ordering and allocation phase. Then, the algorithm proceeds by scanning the tasks in order and assigning them to the appropriate processor(s) one by one, assuming that the processors run at their maximum frequency. Then, the obtained allocation is passed further into a frequency allocation algorithm to determine the operational frequencies for the processors.

\subsection{Task ordering}
The up-rank values, which are quite an effective way to order tasks to minimize makespan \cite{topcuoglu2002performance}, is used for ordering the tasks as defined in Equation No. \ref{eqn13}.

\begin{equation} \label{eqn13}
    \scriptsize
    urv(v_i)= 
\begin{cases}
    \frac{1}{m} \cdot \sum_{k=1}^{m}  T_{exec}[v_i, u_k],& \text{if }v_i = v_n \\
     \frac{1}{m} \cdot \sum_{k=1}^{m}  T_{exec}[v_i, u_k] + \\ \max_{v_j \in succ(v_i)} \{w(v_j, v_i) + urv(v_j)\}             & \text{otherwise}
\end{cases}
\end{equation}

The task order denoted by $\rho = \{v_{\rho(1)}, v_{\rho(2)}, \ldots,v_{\rho(n)}\}$ is found by sorting the tasks in decreasing order of their up-rank values.

\subsection{Reliability constraint}
Depending on the problem setting, we define $\mathcal{R}_{v_i, \max}$ and $\mathcal{R}_{\max}$ below:

\begin{equation} \label{eqn14}
    \scriptsize
    \mathcal{R}_{v_i, \max}= 
\begin{cases}
    \mathcal{R}_{v_i, \max}^{\text{non-fault tolerant}},& \text{if } R_{req} \leq \mathcal{R}_{\max}^{\text{non-fault tolerant}} \\
     \mathcal{R}_{v_i, \max}^{\text{fault-tolerant}}              & \text{otherwise}
\end{cases}
\end{equation}

\begin{equation} \label{eqn15}
    \scriptsize
    \mathcal{R}_{\max}= 
\begin{cases}
    \mathcal{R}_{\max}^{\text{non-fault tolerant}},& \text{if } R_{req} \leq \mathcal{R}_{\max}^{\text{non-fault tolerant}} \\
     \mathcal{R}_{\max}^{\text{fault-tolerant}}              & \text{otherwise}
\end{cases}
\end{equation}

Given a reliability constraint $\mathcal{R}_{req} \leq \mathcal{R}_{\max}$, we wish to allocate processor(s) to a task. For each task $v_i$, define the reliability bound denoted by $\mathcal{R}_{v_i, bound}$ as:

\begin{equation} \label{eqn16}
\mathcal{R}_{v_i, bound} = \mathcal{R}_{req} ^{\log_{\mathcal{R}_{\max}} \mathcal{R}_{v_i, \max}}
\end{equation}

As mentioned earlier, assume that each processor is running at it's maximum frequency.

\begin{lemma} \label{lem1}
If each task satisfies its bound value, the overall reliability constraint is satisfied. 
\end{lemma}

\begin{proof}
Since each task satisfies its bound value, we have : \\
$\mathcal{R}_{v_i, \textbf{k}[i]}^{\textbf{f}[i]} \geq \mathcal{R}_{v_i, bound}, \forall i \in [n]$.

Multiplying the above inequalities, we get:
\begin{multline} \label{eqn17}
\mathcal{R}(\textbf{k}, \textbf{f}) = \prod_{i=1}^{n} \mathcal{R}_{v_i, \textbf{k}[i]}^{\textbf{f}[i]} \geq \prod_{i=1}^{n} \mathcal{R}_{v_i, bound} = \prod_{i=1}^{n} \mathcal{R}_{req} ^{\log_{\mathcal{R}_{\max}} \mathcal{R}_{v_i, \max}} \\ = \mathcal{R}_{req} ^ {\sum_{i=1}^{n} \log_{\mathcal{R}_{\max}} \mathcal{R}_{v_i, \max}} 
 = \mathcal{R}_{req} ^ {\log_{\mathcal{R}_{\max}} \prod_{i=1}^{n} \mathcal{R}_{v_i, \max}} = \mathcal{R}_{req} ^ {\log_{\mathcal{R}_{\max}} \mathcal{R}_{\max}} = \mathcal{R}_{req}
\end{multline}
\end{proof}

Assume that till now, we have finished allocating processors up to task $v_{\rho(i-1)}$, and now we wish to allocate processor(s) for task $v_{\rho(i)}$. Since tasks up to $v_{\rho(i-1)}$ have already been allocated, we know their reliability values. Further assume that the reliability values for tasks $v_{\rho(i+1)}, v_{\rho(i+2)} ,\ldots,v_{\rho(n)}$ are their bound values as given by Equation No. \ref{eqn16}. Then the target reliability for $v_{\rho(i)}$ denoted by $\mathcal{R}_{v_{\rho(i)}, target}$ must be \emph{s.t.} the overall reliability constraint for the task-graph is achieved as illustrated by the below equations using Equation No. \ref{eqn17} in addition:

\begin{equation} \label{eqn18}
\prod_{j=1}^{i-1} \mathcal{R}_{v_{\rho(j)}, \textbf{k}[\rho(j)]}^{\textbf{f}[\rho(j)]} \cdot \mathcal{R}_{v_{\rho(i)}, target} \cdot \prod_{j=i+1}^{n} \mathcal{R}_{v_{\rho(j)}, bound} = \mathcal{R}_{req}
\end{equation}

\begin{equation} \label{eqn19}
\scriptsize
\mathcal{R}_{v_{\rho(i)}, target} = \frac{\mathcal{R}_{req}}{\prod_{j=1}^{i-1} \mathcal{R}_{v_{\rho(j)}, \textbf{k}[\rho(j)]}^{\textbf{f}[\rho(j)]} \cdot \prod_{j=i+1}^{n} \mathcal{R}_{v_{\rho(j)}, bound}} = \frac{\prod_{j=1}^{i} \mathcal{R}_{v_{\rho(j)}, bound}}{\prod_{j=1}^{i-1} \mathcal{R}_{v_{\rho(j)}, \textbf{k}[\rho(j)]}^{\textbf{f}[\rho(j)]}}
\end{equation}

\begin{equation} \label{eqn20}
\scriptsize
\mathcal{R}_{req} \leq \mathcal{R}_{\max} \implies \mathcal{R}_{v_i, bound} = \mathcal{R}_{req} ^{\log_{\mathcal{R}_{\max}} R_{v_i, \max}} \leq \mathcal{R}_{\max} ^{\log_{\mathcal{R}_{\max}} \mathcal{R}_{v_i, \max}} = \mathcal{R}_{v_i, \max}.
\end{equation}


\begin{lemma} \label{lem2}
For each $i \in [n]$, the below equations hold: 
\renewcommand{\theenumi}{\roman{enumi}}%
\begin{enumerate}
\item $\mathcal{R}_{v_{\rho(i)}, target} \leq \mathcal{R}_{v_{\rho(i)}, bound}$
\item $\prod_{j=1}^{i} \mathcal{R}_{v_{\rho(j)}, \textbf{k}[\rho(j)]}^{\textbf{f}[\rho(j)]} \geq \prod_{j=1}^{i} \mathcal{R}_{v_{{\rho(j)}, bound}}$
\end{enumerate}
\end{lemma}

\begin{proof}
The proof proceeds by using induction on i. \\
\textbf{Base Case:} For $i= 1$ using Equation No. \ref{eqn19} we get: $\mathcal{R}_{v_{\rho(i)}, target} = \mathcal{R}_{v_{\rho(i)}, bound}$ which completes (i). By Equation No. \ref{eqn20}, $\mathcal{R}_{v_i, bound} \leq \mathcal{R}_{v_i, \max} \implies \mathcal{R}_{v_{\rho(i)}, target} \leq \mathcal{R}_{v_i, \max}$. Hence there exists a set of processor(s) that can achieve at least $\mathcal{R}_{v_{\rho(i)}, target}$. We choose a set of processors $\{u_{ij} : j \in [i_k]\}$ \emph{s.t.} $\mathcal{R}_{v_{\rho(i)}, \textbf{k}[\rho(i)]}^{\textbf{f}[\rho(i)]} \geq \mathcal{R}_{v_{\rho(i)}, target} = \mathcal{R}_{v_{\rho(i)}, bound}$, completing (ii). \\

\textbf{Inductive Step:} For $i \geq 2$, by our inductive hypothesis we have: $\prod_{j=1}^{i-1} \mathcal{R}_{v_{\rho(j)}, \textbf{k}[\rho(j)]}^{\textbf{f}[\rho(j)]} \geq \prod_{j=1}^{i-1} \mathcal{R}_{v_{{\rho(j)}, bound}}$. Using Equation No. \ref{eqn19}, \\
$\mathcal{R}_{v_{\rho(i)}, target} = \frac{\prod_{j=1}^{i} \mathcal{R}_{v_{\rho(j)}, bound}}{\prod_{j=1}^{i-1} \mathcal{R}_{v_{\rho(j)}, \textbf{k}[\rho(j)]}^{\textbf{f}[\rho(j)]}} \leq \frac{\prod_{j=1}^{i} \mathcal{R}_{v_{\rho(j)}, bound}}{\prod_{j=1}^{i-1} \mathcal{R}_{v_{{\rho(j)}, bound}}} = \mathcal{R}_{v_{\rho(i)}, bound}$ \\
This completes (i). As seen above, since $\mathcal{R}_{v_{\rho(i)}, target} \leq \mathcal{R}_{v_{\rho(i)}, bound}$, hence there exists a set of processor(s) which will satisfy the reliability target. By choosing any such set $\{u_{ij} : j \in [i_k]\}$, we get: \\ $\mathcal{R}_{v_{\rho(i)}, \textbf{k}[\rho(i)]}^{\textbf{f}[\rho(i)]} \geq \mathcal{R}_{v_{\rho(i)}, target} = \frac{\prod_{j=1}^{i} \mathcal{R}_{v_{\rho(j)}, bound}}{\prod_{j=1}^{i-1} \mathcal{R}_{v_{\rho(j)}, \textbf{k}[\rho(j)]}^{\textbf{f}[\rho(j)]}} \\ \implies \mathcal{R}_{v_{\rho(i)}, \textbf{k}[\rho(i)]}^{\textbf{f}[\rho(i)]} \cdot \prod_{j=1}^{i-1} \mathcal{R}_{v_{\rho(j)}, \textbf{k}[\rho(j)]}^{\textbf{f}[\rho(j)]} \geq \prod_{j=1}^{i} \mathcal{R}_{v_{\rho(j)}, bound} \\ \implies \prod_{j=1}^{i} \mathcal{R}_{v_{\rho(j)}, \textbf{k}[\rho(j)]}^{\textbf{f}[\rho(j)]} \geq \prod_{j=1}^{i} \mathcal{R}_{v_{{\rho(j)}, bound}}$, completing (ii).
\end{proof}

\begin{theorem} \label{thm1}
Any given reliability constraint $\mathcal{R}_{req} \leq \mathcal{R_{\max}}$ can be satisfied. 
\begin{proof}
Putting i = n in Equation (ii) of Lemma \ref{lem2}, gives us:\\
$\mathcal{R}(\textbf{k}, \textbf{f}) = \prod_{j=1}^{n} \mathcal{R}_{v_{\rho(j)}, \textbf{k}[\rho(j)]}^{\textbf{f}[\rho(j)]} \geq \prod_{j=1}^{n} \mathcal{R}_{v_{{\rho(j)}, bound}} = \mathcal{R}_{req}$ (From Equation No. \ref{eqn17}). This proves that our strategy can achieve the given reliability constraint.
\end{proof}
\end{theorem}

If $\mathcal{R}_{v_i, \max}, \mathcal{R}_{\max}  \approx $ 1 (especially in the fault-tolerant setting), computing the logarithm in Equation No. \ref{eqn16} becomes infeasible, so we set the bound value to: $\mathcal{R}_{v_i, bound} = \mathcal{R}_{req}^{1/n}$.  

\subsection{Processor allocation}
\begin{itemize}
\item \textbf{Non-Fault Tolerant Setting:} 
We propose an allocation policy based on the finish time, energy consumption, and execution time of the tasks while satisfying the reliability target. After a task has been allocated to a processor, its time requirement is divided into two parts:

\begin{itemize}
\item Idle Time: In this case, the task is waiting to receive the output from its predecessor task(s).
\item Execution Time: In this case, the task is getting executed on the processor.
\end{itemize}
Combining the above two times, we get the wait time of a task defined using Equation No. \ref{eqn21}.

\begin{equation} \label{eqn21}
\scriptsize
    W(v_i)= 
\begin{cases}
    \frac{1}{m} \cdot \sum_{k=1}^{m}  T_{exec}[v_i, u_k] & \text{if }v_i = v_1 \\
     \max_{v_j \in pred(v_i)} \{w(v_j, v_i)\} + \frac{1}{m} \cdot \sum_{k=1}^{m}  T_{exec}[v_i, u_k]             & \text{otherwise}
\end{cases}
\end{equation}

The first part of Equation No. \ref{eqn21} gives an idea of how much time a task has to sit idle on a processor till output from all its predecessor(s) has reached it, while the second part is the average execution time over all processors. Tasks with longer wait times are prone to be on their allocated processor for longer. Thus, it is in favor of a shorter makespan that such tasks be allocated to processors with lower finish times. We sort the tasks based on decreasing order of their wait times in an array $W$. Then, we allocate the first $\ell(0 \leq \ell \leq n)$ tasks based on a weighted normalized linear combination of the task's finish time with weight $\alpha$ and its execution time with weight $(1-\alpha)$. We include execution time in our metric because a lower execution time reduces finish time and energy consumption, as seen in Equation No. \ref{eqn1}, \ref{eqn6} and increases reliability as seen from Equation No. \ref{eqn8}. The remaining tasks are allocated based on their energy consumption. The normalizing of values is done using Min-Max Normalization. Our final allocation algorithm MERT is presented in Algorithm 1.

\begin{algorithm}[!htb]
\caption{MERT scheduling algorithm}
	\KwIn{Task graph, processor parameters, reliability constraint}
    \KwOut{Task-processor Mapping Vector}
	$\text{Compute the ordering } \rho \text{ using Equation No. }  \ref{eqn13}$. \;
	$\text{Calculate } W(v_i) \text{ according to Equation No. } \ref{eqn21}$\; 
	
Sort the tasks based on the waiting time value and break ties based on the ordering in $\rho$\;	
	\For{$i \in [n]$}{
		Calculate the reliability target $\mathcal{R}_{v_{\rho(i)}, target}$ according to Equation No. \ref{eqn19}\;
		Compute the set $\delta$ of processors satisfying the reliability target \emph{i.e.} $\delta = \{u_k :\mathcal{R}_{v_{\rho(i)}, u_k} ^ {f_{u_k, \max}} \geq \mathcal{R}_{v_{\rho(i)}, target}\}$\;
	    \If{$v_{\rho(i)} \in  [W[1], W[\ell]]$} {
	        Assign $v_{\rho(i)}$ to processor satisfying $ u^{*}_{k} \longleftarrow \underset{u_k \in \delta}{\min} \ \{ \alpha \cdot \frac{T_f[v_{\rho(i)}, u_k] - T_f[v_{\rho(i)}, \min]}{T_f[v_{\rho(i)}, \max] - T_f[v_{\rho(i)}, \min]} + (1-\alpha) \cdot \frac{T_{exec}[v_{\rho(i)}, u_k] - T_{exec}[v_{\rho(i)}, \min]}{T_{exec}[v_{\rho(i)}, \max] - T_{exec}[v_{\rho(i)}, \min]} \}$ \;
	    }
	    \Else {
	         Assign $v_{\rho(i)}$ to processor satisfying $u^{*}_{k} \longleftarrow  \underset{u_k \in \delta}{\min}   \ E_{v_{\rho(i)}, u_k}^{f_{u_k, \max}} $ \;
	    }
	    $\textbf{k}[\rho(i)] \longleftarrow \ u^{*}_{k}$\;
	}
	\KwRet $\textbf{k}$\;
	\label{Algo:1}
\end{algorithm}

In Algorithm 1, $T_f[v_{\rho(i)}, \min], T_f[v_{\rho(i)}, \max], T_{exec}[v_{\rho(i)}, \min]$, $T_{exec}[v_{\rho(i)}, \max]$ denote the minimum and maximum values of finish and execution times respectively of executing $v_{\rho(i)}$ over all processors.
\par \emph{Frequency Allocation:} We use the SOEA algorithm, which gives optimum energy consumption for a given reliability constraint and processor allocation \cite{huang2020dynamic}.
\\
\par \emph{Complexity Analysis}
\begin{itemize}
    \item For calculating the bound values for each task, we need to find $\mathcal{R}_{v_i, \max}$ over all processors, as seen from Equation No. \ref{eqn16}, taking $\mathcal{O}(m)$ time. Hence for all tasks, a total of $\mathcal{O}(m \cdot n)$ time is required.
    \item Calculating the up-rank values and wait times for a task $v_i$ takes $\mathcal{O}(m+n)$ time as seen from Equation No. \ref{eqn13}, \ref{eqn21}. Hence, the total time required for all the tasks will be $\mathcal{O}(n \cdot (m + n))$. Then, sorting takes additional $\mathcal{O}(n \cdot \log n)$.
    \item For each task in $\rho$, calculating the target value takes $\mathcal{O}(1)$ time if we updating the values of product of $\mathcal{R}_{\rho(i), bound}$ and $\mathcal{R}_{v_{\rho(i)}, \textbf{k}[\rho(i)]}^{\textbf{f}[\rho(i)]}$ in Equation No. \ref{eqn19} at each step. Finding the set of processors $\delta$ for which the reliability target is satisfied takes $\mathcal{O}(m)$ time. For each processor, finding the finish time takes $\mathcal{O}(n)$ time, as seen from Equation No. \ref{eqn1}, \ref{eqn3}. Determining the execution time and energy takes $\mathcal{O}(1)$ time. The total time taken over all tasks and processors will be $\mathcal{O}(n^2 \cdot m)$.
    \item The complexity of SOEA is $\mathcal{O}(\log(L_o/\epsilon) \cdot m \cdot \log(L_f/\epsilon) \cdot n)$, where $L_o, L_f$ relate to the frequency range under consideration and $\epsilon$ denotes the accuracy of the reliability constraint \cite{huang2020dynamic}.
    \item Finally, by varying the constants $\ell$, $\alpha$ the overall time complexity of MERT+SOEA is $\mathcal{O}(L_{\alpha} \cdot n^2 \cdot m (n + \log(L_o/\epsilon) \cdot \log(L_f/\epsilon)))$, where $L_{\alpha}$ is the number of iterations over $\alpha$. 
\end{itemize}

\item \textbf{Fault-Tolerant Setting:} . We adopt the strategy of sorting the processors in increasing order of their energy consumption for task $v_i$. And then, pick the processors one-by-one until the reliability target is achieved. Our algorithm EAFTS is presented in Algorithm 2.

\begin{algorithm}[!htb]
\caption{EAFTS scheduling algorithm}
	\KwIn{Task graph, processor parameters, reliability constraint}
    \KwOut{Task-processor Mapping Vector}
	$\text{Compute the ordering } \rho \text{ using Equation No. }  \ref{eqn13}$. \;
	\For{$i \in [n]$}{
		Calculate the reliability target $\mathcal{R}_{v_{\rho(i)}, target}$ according to Equation No. \ref{eqn19}\;
		Sort the processors $u_k$ in increasing order of $E_{v_{\rho(i)}, u_k} ^ {f_{u_k, \max}}$ in list $L = \{u_{ij} : j \in [m]\}$, break ties by giving priority to processor with higher reliability value\;
		$\textbf{k}[\rho(i)] \longleftarrow \{\}$\;
		\For{$j \in [m]$}{
		    Add $L[j]$ to $\textbf{k}[\rho(i)]$\;
		    Calculate the reliability $\mathcal{R}_{v_{\rho(i)}, \textbf{k}[\rho(i)]}^{\textbf{f}[\rho(i)]}$ according to Equation No. \ref{eqn9}\;
		    \If{$\mathcal{R}_{v_{\rho(i)}, \textbf{k}[\rho(i)]}^{\textbf{f}[\rho(i)]} \geq \mathcal{R}_{v_{\rho(i)}, target}$}{
		        \textbf{break} \;
		    }
		}
	}
	\KwRet $\textbf{k}$\;
	\label{Algo:2}
\end{algorithm}

\par \emph{Frequency Allocation:} Given a processor allocation and reliability constraint, we proceed to find the operational frequencies for each processor in this sub-section. For a processor $u_{ij}$ running at some frequency $f_{ij}$, we try to find the frequency $f_{ij}^{'}$ \emph{s.t.} the overall reliability is brought down to exactly $R_{req}$ by keeping all the other frequencies the same. Next, we find the decrease in energy obtained by doing so. The processor having the maximum decrease in energy is selected, and the corresponding frequency is set as the mean of $f_{ij}$ and $f_{ij}^{'}$. The algorithm termed Frequency Allocation (FA) is presented in Algorithm 3, where $\zeta$ denotes the accuracy level set to $10^{-5}$. \\ \\

Observe that by setting $\mathcal{R}(\textbf{k}, \textbf{f}) = \mathcal{R}_{req}$ in Equation No. \ref{eqn10} and subsequent rearranging gives us:
\begin{equation} \label{eqn22}
\scriptsize
\mathcal{R}_{v_i, u_j}^{f_{ij}^{'}} = 1 - \frac{1 - \frac{R_{req}}{\prod_{t \in [n], t \neq i} \mathcal{R}_{v_t, \textbf{k}[t]}^{\textbf{f}[t]}}}{\prod_{r \in [i_k], r \neq j}(1-\mathcal{R}_{v_i, u_{ir}}^{f_{ir}})} = h(\mathcal{R}_{req})
\end{equation}

Using Equation No. \ref{eqn8}, we get:
$\mathcal{R}_{v_i, u_{ij}}^{f_{ij}^{'}} = e^{-\lambda_{u_{ij}}(f_{ij}^{'}) \cdot \frac{T_{exec}[v_i, u_{ij}]}{f_{ij}^{'}}} = h(\mathcal{R}_{req}) \implies \frac{\lambda_{u_{ij}}(f_{ij}^{'})}{f_{ij}^{'}} = -\frac{\log (h(\mathcal{R}_{req}))}{T_{exec}[v_i, u_{ij}]}$ \\

For a given processor frequency pair $(u_{ij}, f_{ij})$, observe that the right side of the above equation is a constant and  $\lambda/f$ is a decreasing function of $f$ as seen from Equation No. \ref{eqn7}, so the new frequency $f_{ij}^{'}$ can be obtained by applying a binary search over the frequency range : $[f_{u_{ij}, \min}, f_{u_{ij}, \max}]$. \\

\begin{algorithm}[!htb]
\caption{FA algorithm}
	\KwIn{Task graph, processor parameters, processor allocation (\textbf{k}), reliability constraint}
    \KwOut{Frequency Allocation Vector}
	Initialize \textbf{f} to maximum frequency values\;
	\While{$\lvert \mathcal{R}(\textbf{k}, \textbf{f}) - \mathcal{R}_{req} \rvert \geq \zeta$} {
    	\For{$i \in [n], j \in [i_k]$}{
    	    Calculate $f_{ij}^{'}$ \emph{s.t.} $\mathcal{R}(\textbf{k}, \textbf{f}) = \mathcal{R}_{req}$ using Equation No. \ref{eqn22} and applying binary search\; 
    	    Calculate the corresponding energy $E_{v_i, u_{ij}}^{f_{ij}^{'}}$  using Equation No. \ref{eqn6} \;
    	}
    	Find the task-processor pair$(u_{ij}, f_{ij})$ with least energy\;
    	Set $f_{ij} \longleftarrow \frac{(f_{ij} + f_{ij}^{'})}{2}$ \;
    }
	\KwRet $\textbf{f}$\;
	\label{Algo:3}
\end{algorithm}

\par \emph{Complexity Analysis: }
\begin{itemize}
    \item EAFTS:
    \begin{itemize}
        \item As before, the time for calculating up-rank values and reliability bound values take $\mathcal{O}(n \cdot (m + n))$ time.
        \item For each task, ordering the processors in increasing order of energy values takes $\mathcal{O}(m \cdot \log m)$ time. After that, at most $\mathcal{O}(m)$ time is taken to satisfy the reliability goal of the task. Hence, overall tasks $\mathcal{O}(n \cdot m \cdot \log m)$ time is needed. Overall tasks, $\mathcal{O}(n \cdot (n + m \cdot \log m))$ time is required.
    \end{itemize}
    \item FA:
    \begin{itemize}
        \item In each iteration of the while loop in Line No. 2 of Algorithm 3, at least one processor-frequency pair is updated. Since there are n tasks and each has at most m processors allocated to it, the while loop runs at most $\mathcal{O}(\lvert  F \rvert \cdot m \cdot n)$ times, where $\lvert F \rvert$ denotes the frequency range under consideration.
        \item For each processor-frequency pair (for loop at Line No. 3), calculating the frequency value $f_{ij}^{'}$ takes at most $\mathcal{O}(\log \lvert F \rvert)$ time (line No. 4). Hence, overall task-processor pairs a total of $\mathcal{O}(\log \lvert F \rvert \cdot m \cdot n)$ time is taken.
        \item Overall $\mathcal{O}(\lvert F \rvert \cdot \log \lvert F \rvert \cdot m^2 \cdot n^2)$ time is required.
    \end{itemize}
\end{itemize}
Hence, EAFTS+FA has a time complexity of $\mathcal{O}(n \cdot (n + m \cdot \log m) + \lvert F \rvert \cdot \log \lvert F \rvert \cdot m^2 \cdot n^2)$.
\end{itemize}

\section{Experimental Evaluation} \label{Sec:EXP}
In this section, we describe the experimental evaluation of the proposed solution approach. Initially, we describe the task graphs used in this study.
\subsection{Task Graphs}
We perform the experimental evaluation of our proposed scheduling algorithm on the below workflows used widely in literature for comparison \cite{huang2020dynamic}, \cite{xie2017energy}.
\begin{itemize}
    \item \textbf{Fast Fourier Transform (FFT)}: FFT applications exhibit a high degree of parallelism. For any given positive integer $\rho$, the number of nodes in the task graph can be given by the following equation: $n=(2+\rho) \times 2^{\rho} - 1$ \cite{topcuoglu2002performance}. We consider task graphs with $\rho = 5$ or $n = 223$ \cite{huang2020dynamic}.
    \item \textbf{Gaussian Elimination (GE)}: Compared to FFT, GE applications exhibit a low degree of parallelism. For a given positive integer $\rho$, the number of tasks can be given by the following equation: $n = (\rho^2 + \rho -2)/2$. We consider task graphs with $\rho = 20$ giving $n = 209$ \cite{huang2020dynamic}.
\end{itemize}

\subsection{Methods Compared}
We compare the performance of the proposed scheduling algorithm with the following existing solutions from the literature:
\begin{itemize}
    \item \textbf{Non-Fault Tolerant Setting:} 
    \begin{enumerate}
        \item \textbf{Out Degree Scheduling (ODS)} \cite{huang2020dynamic}: This is one of the state-of-the-art algorithms for workflow scheduling in the multiprocessor system that takes makespan, energy consumption, and reliability into consideration. ODS has the same time complexity as MERT.
        
        \item \textbf{Energy Efficient Scheduling with Reliability Goal (ESRG), } \cite{xie2017energy}: This is another state-of-the-art algorithm that minimizes energy consumption under reliability constraint. It determines both processor and frequency allocation together. Many existing studies have also compared their results with this method.
        
        \item \textbf{Minimizing Resource Consumption Cost with Reliability Goal (MRCRG)} \cite{xie2016resource}: This is another state-of-the-art algorithm that minimizes resource cost under reliability constraint. As in \cite{xie2017energy}, we can modify the algorithm to reduce energy consumption instead of resource cost and name the algorithm Minimizing Energy Consumption with Reliability Goal (MECRG).
        
        \item \textbf{Maximum Reliability (MR)} \cite{huang2020dynamic}, \cite{xie2017minimizing}: As the name suggests, this algorithm assigns each task to the processor that leads to the maximum reliability($\mathcal{R}_{v_i, \max}^{\text{non-fault tolerant}}$) without considering makespan or energy consumption.
    \end{enumerate}
    \item \textbf{Fault Tolerant Setting:} 
    \begin{enumerate}
        \item \textbf{Energy-Efficient Fault-Tolerant Scheduling (EFSRG)} \cite{xie2017energy}: This is a state-of-the-art algorithm that minimizes energy consumption under reliability constraint considering fault tolerance. It determines both processor and frequency allocation together.
        \item \textbf{MaxRe} \cite{zhao2010fault}: This is a state-of-the-art algorithm that minimizes the number of active replicas under a reliability constraint by assuming each task should have the same reliability target = $\mathcal{R}_{req}^{1/n}$.
        \item \textbf{RR} \cite{zhao2013reliable}: This is a slight improvement of MaxRe that minimizes the number of active replicas under a reliability constraint by assuming all the unassigned tasks have the same reliability target = $\mathcal{R}_{req}^{1/n}$.
    \end{enumerate}
\end{itemize}

\subsection{Experimental Set Up}
We implement the proposed solution approach on a workbench system with i5 $10^{\text{th}}$ generation processor and 32GB memory in Python 3.8.10. Processor parameters are set to reflect the real-world characteristics, such as Intel Mobile Pentium III and ARM Cortex-A9 as in \cite{huang2020dynamic}. We simulate a 32 fully connected processor system with parameters chosen randomly in the ranges: $P_{u_k} \in [0.4, 0.8]$, $c_{u_k} \in [0.8, 1.3]$, $f \in [0.3, 1.0]$, $ \alpha_{u_k} \in [2.7, 3.0]$, $\lambda_{u_k} \in [10^{-6}, 10^{-5}]$, $d_{u_k} \in [1, 3]$. $\epsilon$ in SOEA is set to $10^{-5}$, and frequencies are adjusted in steps of  0.0001 \cite{huang2020dynamic}. The experiment is run for 30 instances for each constraint, and the average values are reported. The reliability constraint is set to $\eta \cdot \mathcal{R}_{\max}^{\text{non-fault tolerant}}$, where $\eta$ is set from 0.9 to 0.99 in steps of 0.01 and from 0.991 to 0.999 in steps of 0.001 in the non-fault tolerant setting. In the fault-tolerant setting, we set $\eta$ from 1.001 in steps of 0.001 till $\eta \leq \mathcal{R}_{\max}^{\text{fault-tolerant}}/\mathcal{R}_{\max}^{\text{non-fault tolerant}}$. We choose the maximum value of $\eta$, denoted by $\eta_{\max}$, over all the 30 instances. Then, we run the experiment again till we find 30 new instances each of which satisfy $\eta_{\max} \leq \mathcal{R}_{\max}^{\text{fault-tolerant}}/\mathcal{R}_{\max}^{\text{non-fault tolerant}}$ and take the average values over them.
\par The value of $\alpha$ in MERT can be determined empirically, but for simplicity, we set it from 0.0 to 1.0 in steps of 0.1. In general, ODS does not satisfy high-reliability constraints, so MR replaces it in such cases. We use SOEA on top of ODS, MECRG, MR, and MERT for frequency allocation. For a fair comparison of the algorithms, we do the following. Out of ESRG, MRCRG+SOEA, and MR+SOEA, we choose the one with the least energy consumption $E_{best}$. Then out of the family of solutions generated by ODS+SOEA and MERT+SOEA, we choose the ones with the best makespan and having energy consumption $\leq E_{best}$. In general, we expect that MERT+SOEA gives the best makespan while also having the least energy consumption. For the fault-tolerant setting, we use FA on top of MaxRe, RR, and EAFTS for frequency allocation. In this case, all the algorithms are compared in terms of their energy consumption.

\subsection{Experimental Observations and Discussions}

\begin{figure*}[!ht]
\centering
\begin{tabular}{cccc}
\includegraphics[width=0.26\textwidth]{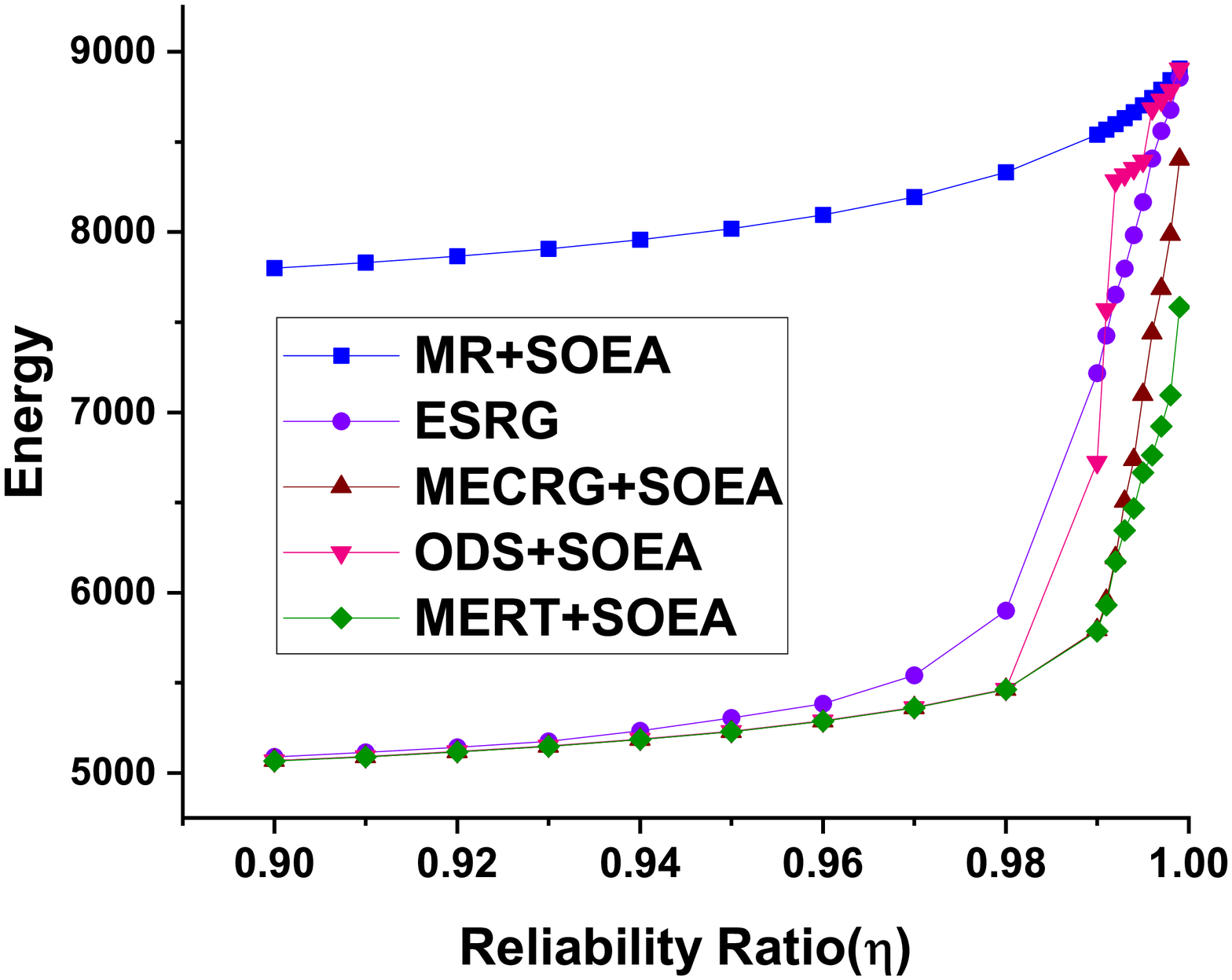} &
\includegraphics[width=0.26\textwidth]{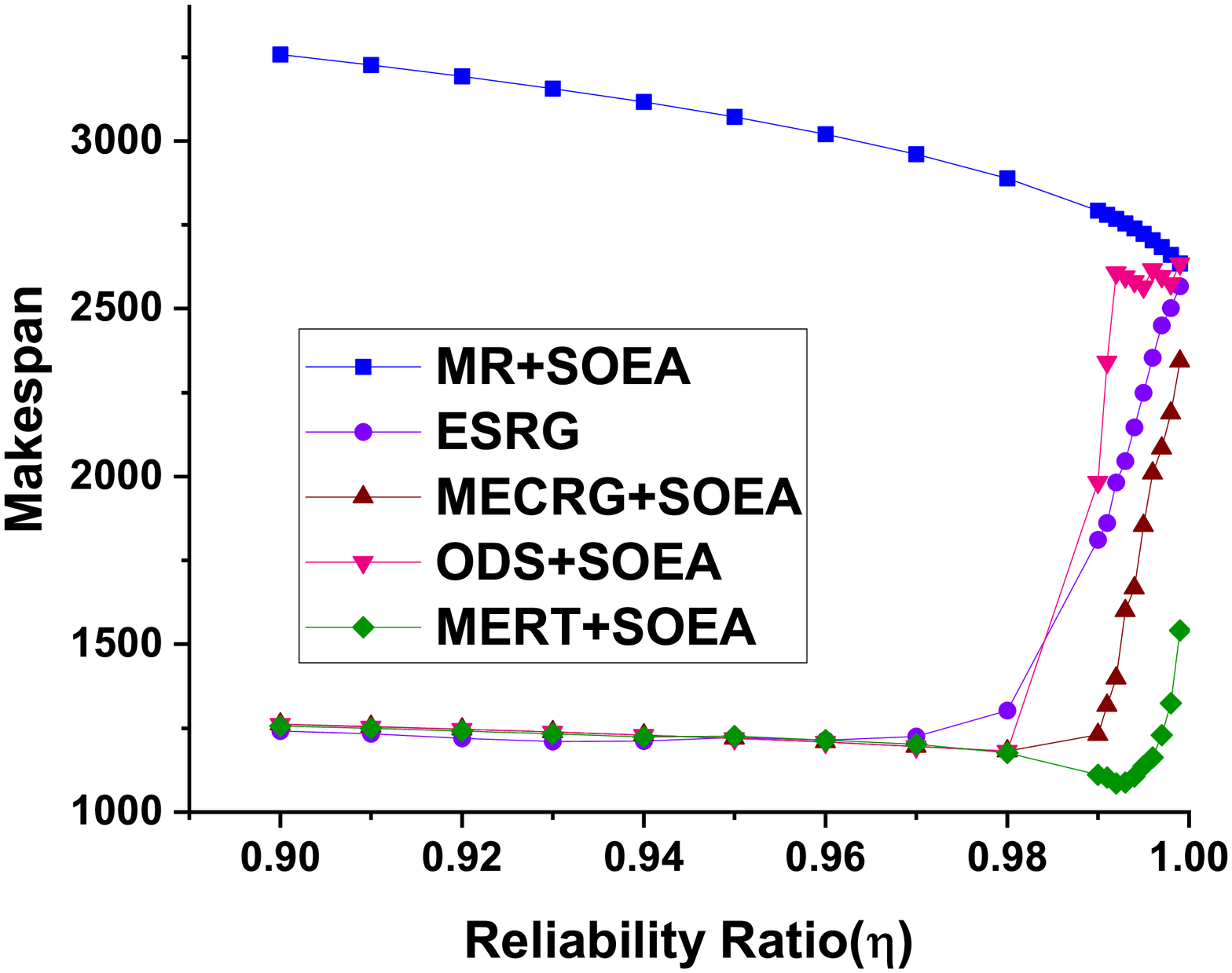} &
\includegraphics[width=0.26\textwidth]{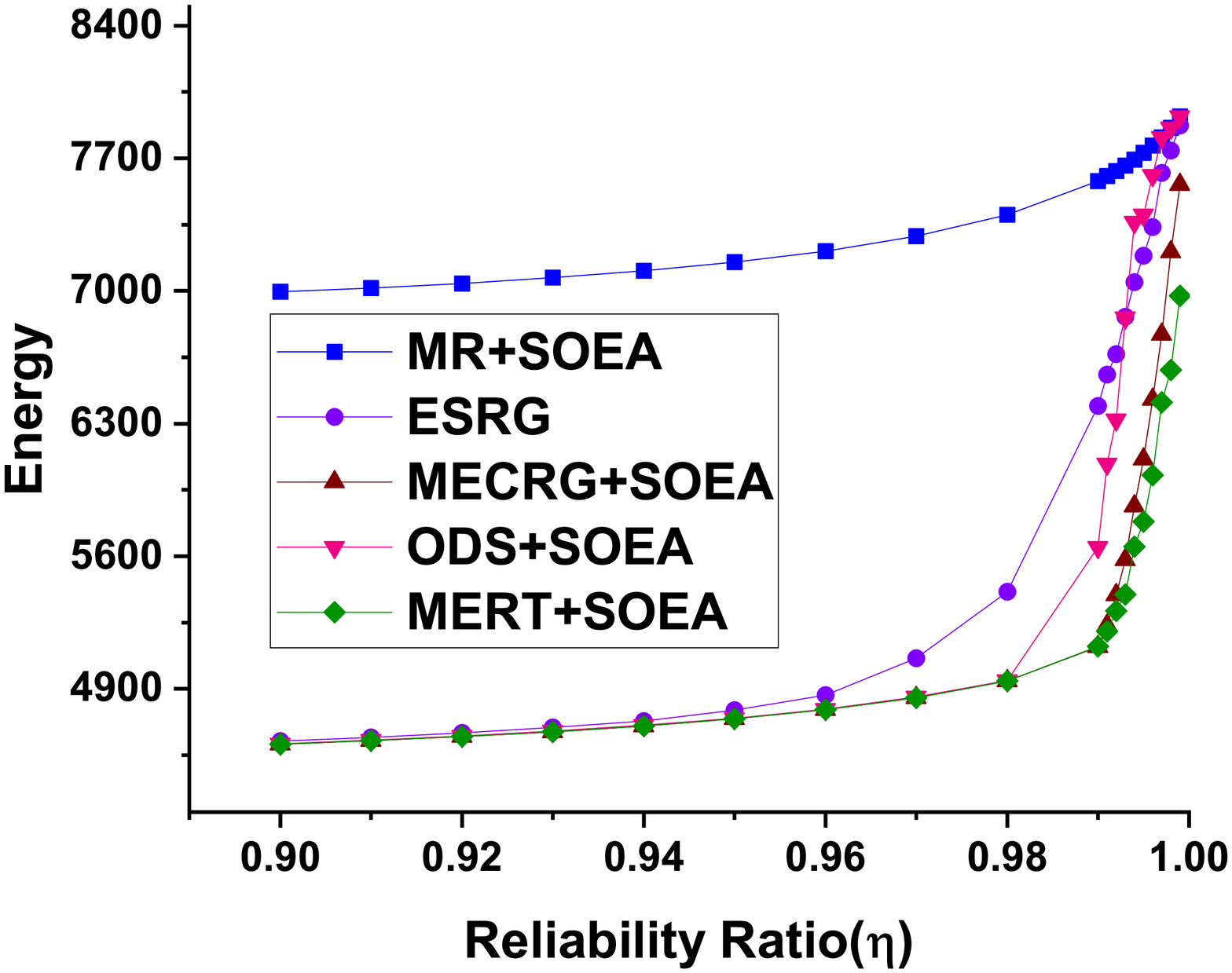} &
\includegraphics[width=0.26\textwidth]{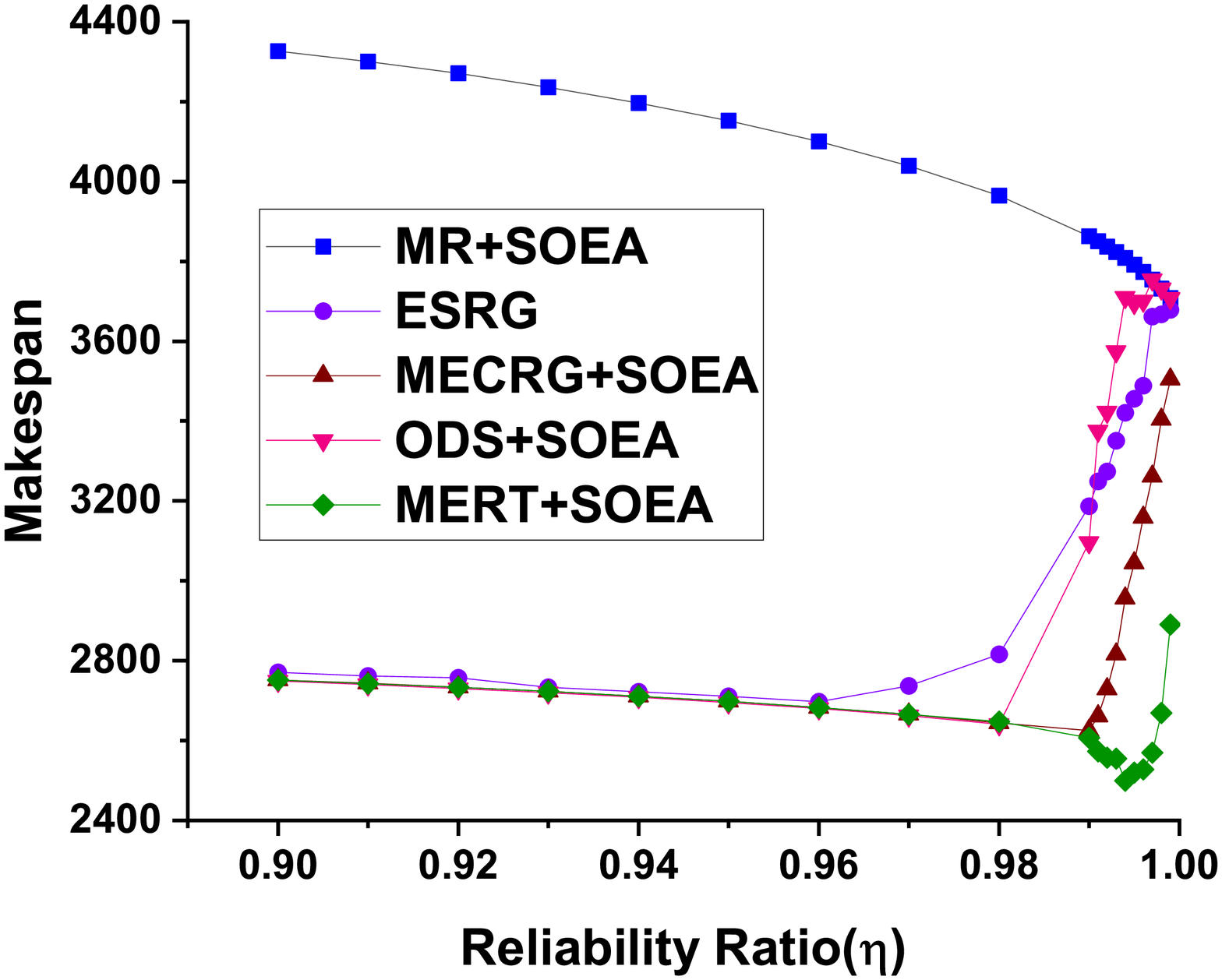} 
\\
(a) & (b) & (c) & (d) \\
\end{tabular}
\caption{Reliability ratio vs. makespan, energy consumption plots for the FFT and GE Workflows in non-fault tolerant setting}
\label{Fig:NFT}
\end{figure*}

\begin{figure*}[h]
\centering
\begin{tabular}{cc}
\includegraphics[scale=0.16]{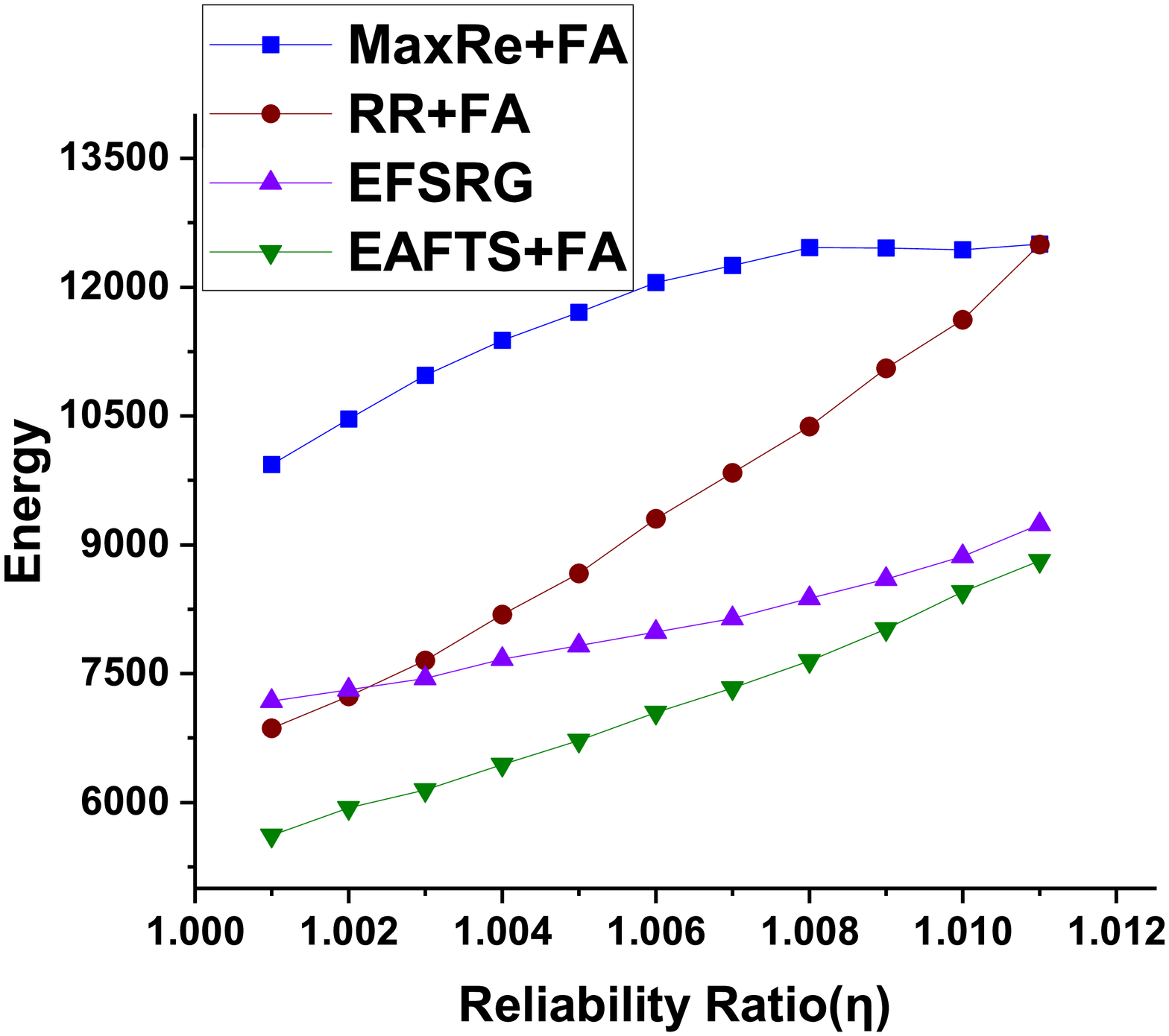} & \includegraphics[scale=0.16]{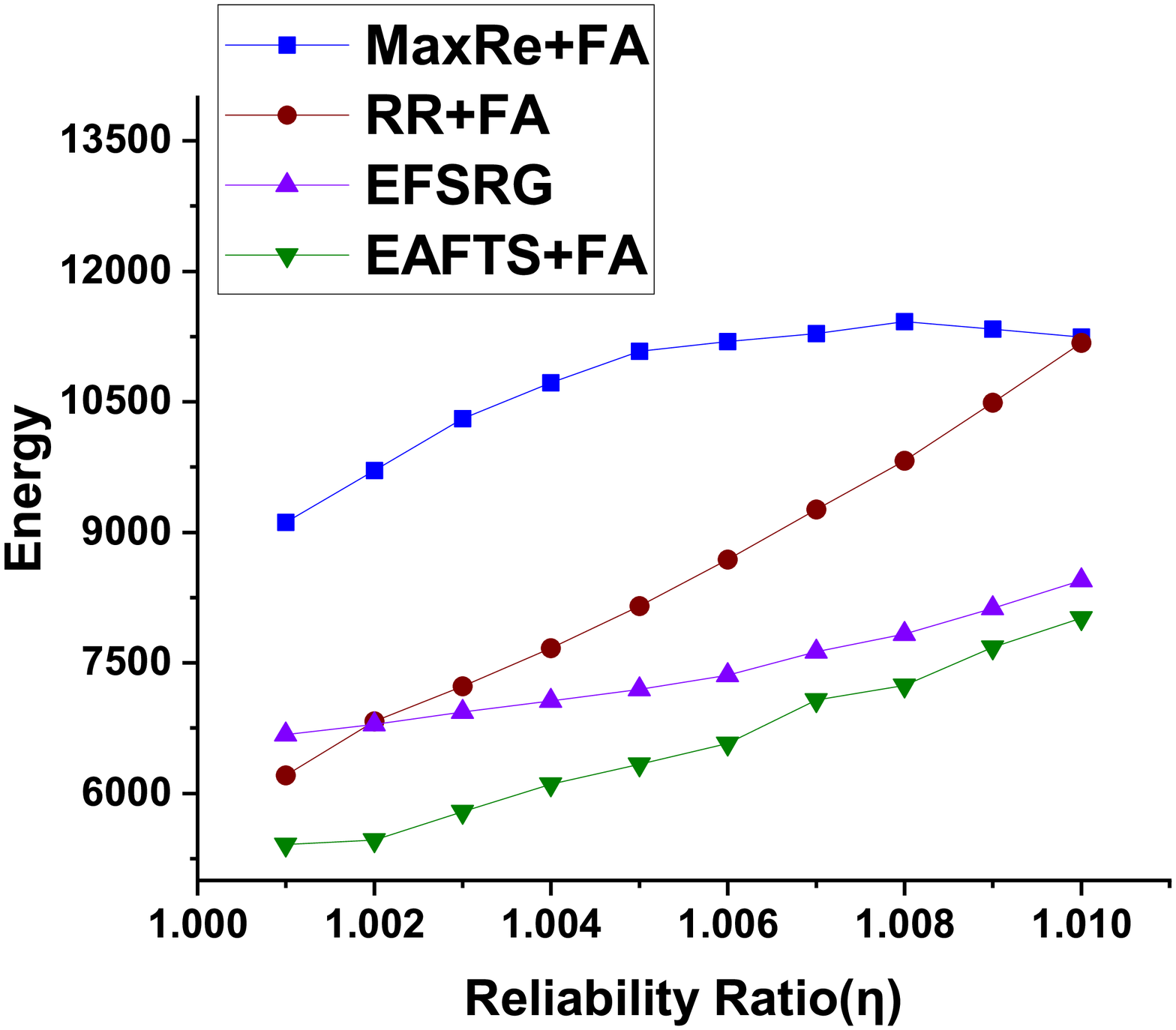} \\
(a) & (b) \\
\end{tabular}
\caption{Reliability ratio vs. energy consumption plots for the FFT and GE Workflows in fault tolerant setting}
\label{Fig:FT}
\end{figure*}

Below we discuss the performance of our proposed algorithms. The values when one algorithm outperforms another are reported on average as percentages, and our methods are compared with the following best method.
\begin{itemize}
    \item \textbf{Non-Fault Tolerant Setting:}
    Fig. \ref{Fig:NFT} (a), (c) show the plots for energy consumption and Fig. \ref{Fig:NFT} (b), (d) show the plots for makespan for FFT and GE workflows respectively. From Fig. \ref{Fig:NFT} (a), we can see that MERT+SOEA achieves the least energy in all cases. Up to $\eta = 0.98$, MERT+SOEA performs similarly to MECRG+SOEA and ODS+SOEA. Up to $\eta = 0.993$, we can see that MERT+SOEA performs similarly to MECRG+SOEA, while ODS+SOEA starts having increased energy consumption. After this, MERT+SOEA outperforms MECRG+SOEA by 7.74\%. From Fig. \ref{Fig:NFT} (b), we observe that up to $\eta = 0.95$, ESRG gives the least makespan performing slightly better than MERT+SOEA, in particular by  1.2\%. Again, up to $\eta = 0.993$, we can see that MERT+SOEA performs similarly to MECRG+SOEA, and ESRG starts having more makespan. After this, MERT+SOEA gives the least makespan outperforming MECRG+SOEA by 37.51\%.
    
    \par From Fig. \ref{Fig:NFT} (c), we can see that MERT+SOEA achieves the least energy in all cases. Up to $\eta = 0.98$, MERT+SOEA performs similarly to MECRG+SOEA and ODS+SOEA. Up to $\eta = 0.993$, we can see that MERT+SOEA performs similarly to MECRG+SOEA. After this, MERT+SOEA outperforms MECRG+SOEA by 5.91\%. From Fig. \ref{Fig:NFT} (d), we can see that MERT+SOEA achieves the least makespan in all cases. Again, up to $\eta = 0.993$, we can see that MERT+SOEA performs similarly to MECRG+SOEA. After this, MERT+SOEA significantly improves over MECRG+SOEA by 17.68\%.
    
    \item \textbf{Fault-Tolerant Setting:}
    \par Fig. \ref{Fig:FT} (a) shows the energy plot for the FFT workflow. In this case, $\eta_{\max} = 1.011$. It is observed that for all values of $\eta$, EAFTS+FA gives the least energy consumption. For $\eta = 1.001, 1.002$, the next best algorithm is RR+FA, followed by EFSRG, and EAFTS+FA outperforms RR+FA by 18.08\%. For $\eta \geq 1.003$, EFSRG performs better than RR+FA by 16.88\%, and EAFTS+FA outperforms EFSRG by 10.15\%.
    
    \par Fig. \ref{Fig:FT} (b) shows the energy plot for the GE workflow. In this case, $\eta_{\max} = 1.01$. Like before, EAFTS+FA gives the least energy consumption for all values of $\eta$. For $\eta = 1.001$, the next best algorithm is RR+FA, followed by EFSRG, and EAFTS+FA outperforms RR+FA by 12.85\%. For $\eta \geq 1.002$, EFSRG performs better than RR+FA by 15.04\%, and EAFTS+FA outperforms EFSRG by 10.55\%.
\end{itemize}

In summary, energy consumption increases with increasing reliability, and makespan decreases. However, for high-reliability constraints, the processors with early finish times may not satisfy the reliability target, leading to an increase in makespan. Our algorithm MERT+SOEA outperforms state-of-the-art algorithms by performing at least as well as them in terms of energy consumption and significantly better in terms of makespan, especially for higher values of $\eta$. In the fault-tolerant setting, our algorithm EAFTS+FA always outperforms the remaining algorithms. The improvement decreases as $\eta$ increases. 
RR always outperforms MaxRe because MaxRe assumes that each task has the same reliability target, whereas RR calculates the reliability target depending on the reliability values of already assigned tasks. Hence, MaxRe needs more replicas to satisfy the reliability target of each task, leading to more energy consumption.

\section{Conclusion and Future Research Directions} \label{Sec:CFA}
This paper proposes MERT and EAFTS allocation algorithms that minimize the total energy consumption under a given reliability constraint in the non-fault tolerant and fault tolerant settings, respectively. Additionally, MERT is designed to account for makespan as well. Next, we propose FA, a frequency allocation algorithm that can be combined with any other processor allocation algorithm in the fault-tolerant setting. All the algorithms are analyzed for their time requirements. Experimental evaluation of the proposed solutions on benchmark task graphs indicates they outperform the state-of-art algorithms. The future study on this problem will focus on developing more efficient solution approaches.

 \bibliographystyle{splncs04}
 \bibliography{Paper}
\end{document}